\documentclass[journal]{IEEEtran}
\usepackage{graphicx}
\usepackage{amssymb}
\usepackage{amsmath}
\usepackage{amsthm}
\usepackage{array}
\usepackage{multirow}
\usepackage{rotating}
\usepackage[ruled]{algorithm2e}
\usepackage{color}

\theoremstyle{plain}
\newtheorem{myTheo}{Theorem}

\theoremstyle{definition}

\theoremstyle{corollary}
\newtheorem{myCor}{Corollary} 

\theoremstyle{remark}
\newtheorem{myRemark}{Remark}

\newtheorem{myAssume}{Assumption}

\newcommand{\B}[1]{\mathbb{#1}}

\makeatletter

\newcommand{\Rmnum}[1]{\expandafter\@slowromancap\romannumeral #1@}
\makeatother

\begin{document}

\title{Testing Scenario Library Generation for Connected and Automated Vehicles, Part I: Methodology}

\author{Shuo~Feng,
	Yiheng~Feng,
	Chunhui~Yu,
	Yi~Zhang,~\IEEEmembership{Member,~IEEE},
	and Henry~X.~Liu,~\IEEEmembership{Member,~IEEE}
	
	\thanks{This work was supported by USDOT Center for Connected and Automated Transportation at the University of Michigan, Ann Arbor. \emph{(Corresponding author: Henry X. Liu)}}
	
	\thanks{S. Feng, Y. Zhang are with the Department of Automation, Tsinghua University, Beijing 100084, China. S. Feng is also a visiting Ph.D. student of the Department of Civil and Environmental Engineering, University of Michigan, Ann Arbor, MI, 48109, USA. (e-mail: s-feng14@mails.tsinghua.edu.cn; zhyi@tsinghua.edu.cn)}
	
	\thanks{Y. Feng is with the University of Michigan Transportation Research Institude, 2901 Baxer Rd, Ann Arbor, MI, 48109, USA. (e-mail: yhfeng@umich.edu)}
	
	\thanks{C. Yu is with the Key Laboratory of Road and Traffic Engineering of the Ministry of Education, Tongji University, 4800 Cao'an Road, Shanghai, China. (e-mail: 13ych@tongji.edu.cn)}
	
	\thanks{H. X. Liu is with the Department of Civil and Environmental Engineering, University of Michigan, Ann Arbor, MI, 48109, USA. (e-mail: henryliu@umich.edu)}
	
}

\markboth{To appear in IEEE Transactions on Intelligent Transaportation Systems}
{Shell \MakeLowercase{\textit{et al.}}: Bare Demo of IEEEtran.cls for IEEE Journals}

\maketitle

\begin{abstract}
Testing and evaluation is a critical step in the development and deployment of connected and automated vehicles (CAVs), and yet there is no systematic framework to generate testing scenario library. This study aims to provide a general framework for the testing scenario library generation (TSLG) problem with different operational design domains (ODDs), CAV models, and performance metrics.  Given an ODD, the testing scenario library is defined as a critical set of scenarios that can be used for CAV test. Each testing scenario is evaluated by  a newly proposed measure, scenario criticality, which can be computed as a combination of maneuver challenge and exposure frequency. To search for critical scenarios, an auxiliary objective function is designed,  and a multi-start optimization method along with seed-filling is applied. Theoretical analysis suggests that the proposed framework can obtain accurate evaluation results with much fewer number of tests, if compared with the on-road test method. In part II of the study, three case studies are investigated to demonstrate the proposed method.  Reinforcement learning based technique is applied to enhance the searching method under high-dimensional scenarios.
\end{abstract}


\IEEEpeerreviewmaketitle

\section{Introduction}

\IEEEPARstart{T}{esting} and evaluation is a critical step in the development and deployment of connected and automated vehicles (CAVs). Testing procedures for human-driven vehicles, such as Federal Motor Vehicle Safety Standards (FMVSS), have been established for a long time. However, current standards only regulate automobile safety-related components, systems, and design features, without consideration of driver performance in completing driving tasks. For CAVs, it is essential to evaluate the ``intelligence'' of the vehicle \cite{li2018artificial}, similar to a driver's license test, which indicates whether a CAV can operate safely and efficiently without human intervention.

Currently, CAV testing and evaluation is mainly conducted via the following methods: simulation test, closed facility test, and on-road test \cite{thorn2018framework}. All three methods have pros and cons. Simulation test is cost-effective, but it is difficult to model exact vehicle dynamics and road environment. On-road test is most realistic, but it is extremely inefficient. A CAV would have to drive hundreds of millions of miles to validate the safety at the level of human-driven vehicles \cite{kalra2016driving}. The underlying reason is that most on-road scenarios are not challenging enough to evaluate the performances of a CAV. For instance, if we want to evaluate the safety performance (e.g., accident rate) of a CAV by analyzing its reaction to red light running vehicles at signalized intersections, it may require the CAV to pass thousands of intersections to accumulate enough accident events, which becomes intractable.

Closed facility test has its unique advantages over the other two methods. It does not require the detailed modeling of vehicle dynamics, which is a must in simulation.   It also provides a more controlled and therefore safer environment for CAV testing than the on-road test method. Moreover, the closed facility test has potential to greatly improve the testing efficiency, i.e., obtain the evaluation results with the same accuracy with fewer number of tests. We should note that, despite the advantages of simulation test and closed facility test, on-road testing is still irreplaceable before deployment. With properly designed scenarios in the simulation and the closed test facility, however, the effort for on-road test can be reduced.

Therefore, the key to exploiting the advantages of either simulation or closed facility test is to generate a testing scenario library for each operational design domain (ODD). The ODD is defined as operation conditions under which a given driving automation system is specifically designed to function \cite{J3016_201806}. Given an ODD, there can exist millions of scenarios with different parameters (e.g., different behaviors of the background vehicles (BVs)).  A testing library is defined as a subset of scenarios that can be used for evaluation of certain pre-defined performance metrics (e.g., safety). Since the library includes more critical scenarios,  testing in a closed facility is usually much more efficient than that on public roads.

In the past few years, increasing research efforts have been made to solve the testing scenario library generation (TSLG)  problem. Generally speaking, the TSLG problem can be disassembled into four closely related research questions: (1)	How to parameterize a testing scenario and define the decision variables? (\emph{Scenario Description}) (2) What are the performance metrics for CAV evaluation? (\emph{Metric Design})	(3)	How to generate a testing scenario library for a specific performance metric? (\emph{Library Generation}) (4) How to evaluate CAVs with the generated library? (\emph{CAV Evaluation}) A brief overview of existing studies will be provided in Section II. To the best of our knowledge, all existing methods have limitations in either ODD types that can be handled (e.g., low-dimensional scenarios only), CAV models (e.g., a specific CAV only), or performance metrics (e.g., safety evaluation only). 

In this paper, a unified framework for TSLG is proposed, as shown in Fig. \ref{fig_Framework}. The four research questions are integrated and solved together in the framework: (1) Decision variables of a scenario are formulated by scenario parameterization considering ODD (Section III.A). (2) Incremental performance metrics are designed, including safety, functionality, mobility, and rider's comfort (Section III.B).  (3) A method is proposed to generate the testing scenario library, including a new criticality definition and an optimization-based searching method for critical scenarios (Section IV). (4) With the generated library, CAVs are evaluated by scenario sampling, CAV testing, and index estimation (Section V). 

\begin{figure*}
	\centering
	\includegraphics[width=0.95\textwidth]{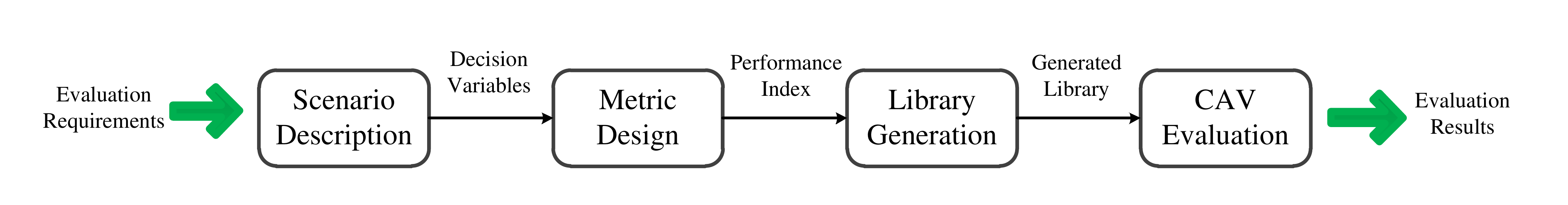}
	\caption{An illustration of the proposed framework to the TSLG problem.} 
	\label{fig_Framework}
\end{figure*}

The library generation is the key step in the entire framework (Section IV). The basic idea is to define the criticality of scenarios and search the set of critical scenarios to construct the library. To evaluate the importance of a scenario, a new definition of criticality is proposed as a combination of maneuver challenge and exposure frequency, as scenarios with higher occurrence probability in the real-world and higher maneuver challenge should have higher priority for CAV evaluation. The maneuver challenge is estimated by a surrogate model of CAVs, whereas the exposure frequency is calculated based on naturalistic driving data. The new definition is fundamentally different from most existing studies, which usually overvalue worst-case scenarios \cite{jung2007worst}\cite{zhao2018accelerated}.  In order to reduce the computational complexity in the process of searching for critical scenarios, an auxiliary objective function is designed to guide the searching direction,  and the seed-fill method is applied to search neighborhood scenarios.

Theoretical analysis in Section VI provides justifications of the proposed method for both evaluation accuracy and efficiency. Specifically, the proposed method obtains unbiased index estimation of performance metrics (i.e., accuracy), and the estimation variance is zero under certain conditions (i.e., efficiency). Based on the theoretical analysis, hyper-parameters (i.e., the threshold of critical scenarios and parameters of sampling policy) can be determined.

This study is divided into two parts. Overall framework, methodologies, and theoretical analysis are presented in this paper. In Part II paper  \cite{feng2019testing}, three case studies are investigated to demonstrate the proposed methodologies. 


\section{Related Work}
\label{sec_Prob}
In this section, a brief overview of related work is provided from the perspectives of the four research areas, i.e., scenario description, metric design, library generation, and CAV evaluation. 

Scenario description focuses on the parameterization of testing scenarios and definition of decision variables. A scenario describes the temporal development among a sequence of scenes, which include snapshots of  the environment (e.g., background vehicles,  road information, and environment conditions) \cite{ulbrich2015defining}. Decision variables in most existing studies are defined by listing all possible influencing factors, which is intractable when the testing scenarios are complex. To reduce the complexity, Li et al. \cite{li2016intelligence}\cite{li2019parallel} decomposed testing scenarios as a series of pre-determined driving tasks, which can be specified by a group of spatial-temporal attributes. Zhou et al. \cite{zhou2017reduced} described a complex testing scenario (e.g., overtaking) by several basic scenarios (e.g., car-following and lane changing) and a set of transition rules.  The PEGASUS project \cite{PEGASUS} proposed a three-level framework to describe testing scenarios, i.e., functional level, logical level, and concrete level. If parameters of the top two levels are pre-determined, then the decision variables include only the parameters of the concrete level.

For performance metrics and related indices for CAV evaluation, most current studies focus on safety only, which is usually assessed by the disengagement rate or the accident rate \cite{Waymo}\cite{zhao2017accelerated}. Although safety is the foundation of all CAV applications, a safe but over-conservative CAV may fail in simple driving tasks. Therefore, functionality, which represents the vehicle's ability to complete driving tasks, should also be included in the evaluation process. Furthermore, mobility and rider's comfort can be considered as higher level requirements. To better evaluate the metrics of CAVs, quantitative indices are desirable. Most existing studies, however, can only obtain qualitative assessment, e.g., ISO 26262 \cite{iso26262} provides four safety integrity levels from A to D. 

Testing scenario library generation is key to CAV test. The most straightforward method is to design a ``test matrix'' based on crash data analysis \cite{najm2013depiction}\cite{carsten2005human}\cite{karabatsou2007priori}, naturalistic driving data (NDD) analysis \cite{roesener2016scenario}\cite{kruber2018unsupervised}, and scenario randomization \cite{khastgir2017test}, as well as similarity analysis \cite{arcuri2010black}\cite{hemmati2010reducing}\cite{shin2010normalized} and coverage analysis \cite{hemmati2010enhanced}\cite{hemmati2010industrial}, which were developed for software verification and validation. However, as the test matrix is pre-determined, a CAV can be specifically trained to achieve good performance in the test, which is problematic for CAV evaluation. Toward addressing this issue,  the worst-case scenario evaluation (WCSE) method was developed with the knowledge of exact CAV dynamics and driving behaviors, which is usually intractable \cite{jung2007worst}. To avoid this problem, a black-box searching method was used to identify testing scenarios by adaptively testing a particular CAV \cite{mullins2018adaptive}. Both the WCSE and black-box searching methods can be used to generate scenarios for a particular CAV only. To construct testing scenarios for generic CAVs, the PEGASUS project \cite{PEGASUS} numerically measured the ``risk'' of all feasible scenarios that was defined by ISO 26262 and selected the risky testing scenarios. 
However, testing scenarios generated using the abovementioned methods may not reflect real-world driving conditions, therefore test results with these scenarios may not represent a CAV's true performance.

For CAV evaluation, most existing methods estimate the accident rate of a CAV using scenarios from NDD, such as naturalistic field operational tests \cite{festa2008festa} and crude Monte Carlo method \cite{yang2010development}\cite{lee2004longitudinal}\cite{de2017assessment}. However, this method is proved inefficient and intractable for even low-dimensional scenarios \cite{kalra2016driving}. To address this problem, Zhao et al. \cite{zhao2017accelerated} introduced importance sampling techniques. Instead of sampling testing scenarios from NDD, an importance function was constructed when conducting CAV testing. However, construction of the importance function remains challenging. The cross entropy method applied in \cite{zhao2017accelerated} was based on adaptively testing a particular CAV, which requires prohibitively huge number of CAV testing for high-dimensional scenarios. As a result, under high-dimensional car-following scenarios \cite{zhao2018accelerated}, the cross entropy method was replaced by a white-box optimization method with the assumption of exact CAV models, which is a huge limitation.

Notwithstanding the related studies, all existing methods have limitations in either ODD types that can be handled (e.g., low-dimensional scenarios only), CAV models (e.g., a specific CAV only), or performance metrics (e.g., safety evaluation only). To the best of our knowledge, no existing studies has integrated all parts of the TSLG problem together and are cable of generating libraries for different scenario types, CAV types, and performance metrics.

\section{Problem Formulation}
In this section, the TSLG problem is analyzed and formulated. In Subsection III.A, decision variables of a testing scenario are defined. Performance metrics for a CAV test, including safety, functionality, mobility, and rider's comfort, are described in Subsection III.B. To quantitatively measure the metrics, the occurring probability of the event of interest is used as the performance indices and described in Subsection III.C. To improve the efficiency of performance index estimation, importance sampling techniques are also introduced in this subsection. As shown in Subsection III.D,  the generation of the testing scenario library is equivalent to the construction of the importance function. Finally, the assumptions made in a CAV test are provided in Subsection III.E. Notations of variables are listed in Table \ref{tab_notation}.

\linespread{1.2}
\begin{table}
	\centering
	\footnotesize
	\setlength{\abovecaptionskip}{1pt}
	\setlength{\belowcaptionskip}{3pt}	
	\caption{Notations of the variables in this paper. }
	\label{tab_notation}
	\begin{tabular}{m{1.5cm}<{\centering}p{6.5cm}}
		\hline
		\multicolumn{1}{c}{\bfseries Variables } &   \multicolumn{1}{c}{ \bfseries Notations}   \\ \hline
		$\theta$  & \multicolumn{1}{m{6.5cm}}{Pre-determined parameters of testing scenarios by operational design domain.}\\ \hline
		$x$  & \multicolumn{1}{m{6.5cm}}{Decision variables of testing scenarios.}\\ \hline
		$A$ & \multicolumn{1}{m{6.5cm}}{Event of interest (e.g., accident) with a CAV model.}\\ \hline
		$S$ & \multicolumn{1}{m{6.5cm}}{Event of interest (e.g., accident) with a surrogate model.}\\ \hline
		$\B{X}$ & \multicolumn{1}{m{6.5cm}}{Feasible set of decision variables.} \\\hline
		$\Phi$ & \multicolumn{1}{m{6.5cm}}{Critical set of decision variables.} \\\hline
		$\gamma$ & \multicolumn{1}{m{6.5cm}}{Criticality threshold of critical scenarios.} \\\hline
		$q(x)$ & \multicolumn{1}{m{6.5cm}}{Importance function.} \\\hline
		$P(S|x, \theta)$ & \multicolumn{1}{m{6.5cm}}{Probability of event $S$ in scenario $(x, \theta)$, i.e., maneuver challenge.} \\\hline
		$P(x|\theta)$ & \multicolumn{1}{m{6.5cm}}{Occurring probability of scenario $(x, \theta)$ on-road, i.e., exposure frequency.} \\\hline
		$V(x|\theta)$ & \multicolumn{1}{m{6.5cm}}{Criticality value of scenario $(x, \theta)$.} \\ \hline
		$N(\B{X}), N(\Phi)$ & \multicolumn{1}{m{6.5cm}}{Total number of  scenarios in the set $\B{X}$, $\Phi$.} \\ \hline
		$\bar{P}(x| \theta)$ &\multicolumn{1}{m{6.5cm}}{Testing probability of scenario $(x, \theta)$.} \\ \hline
		$\bar{P}_1(x| \theta)$ & \multicolumn{1}{m{6.5cm}}{$\bar{P}(x| \theta)$ for greedy sampling policy.} \\\hline
		$\bar{P}_2(x| \theta)$ & \multicolumn{1}{m{6.5cm}}{$\bar{P}(x| \theta)$ for $\epsilon$-greedy sampling policy.} \\\hline
		$\epsilon$ & \multicolumn{1}{m{6.5cm}}{Exploration probability of $\epsilon$-greedy sampling policy.} \\\hline
		$\hat{P}(A|\theta)$ &\multicolumn{1}{m{6.5cm}}{Estimated probability of the event $A$ with pre-determined parameters $\theta$.} \\ \hline 
		$n$ &\multicolumn{1}{m{6.5cm}}{Total number of sampled testing scenarios.} \\ \hline
		$J(x)$ & \multicolumn{1}{m{6.5cm}}{Auxiliary objective function.} \\\hline
		$mnpETTC$ & \multicolumn{1}{m{6.5cm}}{Minimal normalized positive enhanced time-to-collision during testing.} \\\hline
		$R, \dot{R}$ & \multicolumn{1}{m{6.5cm}}{Range and range rate at the cut-in moment between the background vehicle and test CAV.} \\\hline
		$R(t), \dot{R}(t)$ & \multicolumn{1}{m{6.5cm}}{Range and range rate at time $t$ between the background vehicle and test CAV.} \\\hline
		$\omega$ & \multicolumn{1}{m{6.5cm}}{Weight parameter.} \\\hline
		$d(x, \Omega)$ & \multicolumn{1}{m{6.5cm}}{Normalized distance between scenario $x$ and a high exposure frequency zone $\Omega$.} \\\hline
		$W$ & \multicolumn{1}{m{6.5cm}}{Normalization factor.} \\\hline
		$f_A(x)$ & \multicolumn{1}{m{6.5cm}}{Probability of event $A$ in scenario $(x,\theta)$, i.e., $P(A|x,\theta)$.} \\\hline
		$f_S(x)$ & \multicolumn{1}{m{6.5cm}}{Probability of event $S$ in scenario $(x,\theta)$, i.e., $P(S|x,\theta)$.} \\\hline
	\end{tabular}
\end{table}
\linespread{1.0}

\subsection{Decision Variables}
The terms scene and scenario defined in \cite{ulbrich2015defining} are adopted. A scene describes a snapshot of the environment including the scenery and dynamic elements. Scenery includes all geo-spatially stationary elements, which entails metric, semantic, topological, and categorical information about roads and all the subcomponents such as lanes, lane markings, and road surface types. Dynamic elements are those moving or have the ability to move, e.g., pedestrians and vehicles. A scenario describes the temporal development in a sequence of scenes. 

Testing scenarios should be consistent with the ODD \cite{J3016_201806}. Usually, for a given testing scenario, most of its stationary elements are specified by the ODD. ODD also provide constraints for the dynamic elements of a testing scenario. 
In this paper, the parameters determined by the ODD are denoted as $\theta$, e.g., number of lanes, road type, weather conditions, \emph{etc}.
Then the remaining parameters (e.g., behaviors of background vehicles (BVs)) are denoted as a vector of decision variables as
\begin{eqnarray}
x = \left[ x(1), x(2), \cdots, x(d) \right],
\end{eqnarray}
where $d$ denotes the dimensionality of $x$. The feasible set of $x$, i.e., $\B{X}$, is determined by the ODD, e.g., speed range, acceleration range, and perception range. The main task of the TSLG problem is to determine a critical subset $\Phi$ of  $\B{X}$ (i.e., $\Phi \subset \B{X}$), which can be used for CAV evaluation.

Taking the cut-in scenario as an example, the decision variables can be formulated as
\begin{eqnarray}
x = \left[R, \dot{R}\right], x \in \B{X}
\end{eqnarray}
where $R$ and $\dot{R}$ denote the range (i.e., relative distance) and range rate (i.e., relative speed, assuming the speed of the CAV is given) between the BV and the test CAV at the cut-in moment  \cite{PEGASUS}\cite{zhao2017accelerated}. The feasible set $\B{X}$ (i.e., range limit and range rate limit) and the constant parameters $\theta$ (e.g., number of lanes and road type) are determined by the ODD.

\subsection{Performance Metrics}
Performance metrics define what aspects a CAV needs to be evaluated. Most existing studies focus only on safety evaluation, which is necessary but insufficient for a deployable CAV. In this paper, we define the performance metrics to reflect people's incremental expectations towards CAVs, including safety, functionality, mobility, and rider's comfort.

Safety is the foundation of all CAV applications, which is usually assessed by the disengagement rate or the accident rate without human intervention \cite{Waymo}\cite{zhao2017accelerated}. Again, taking the cut-in scenario as an example, a BV changes its lane in front of a CAV in the adjacent lane with a specified realization of decision variables, i.e., cut-in distance and speed difference. Whether an accident (e.g., conflict or crash) happens or not depends on the CAV's response to the BV's action. After a certain number of tests with varying realizations of decision variables, the accident rate of the CAV could be estimated, which is used to indicate the safety performance in the cut-in scenario. 

Functionality is another important performance metric, which is defined by whether a CAV can complete a given task in a specific scenario. Consider a scenario that a CAV needs to make a lane change to the right and exit the highway within a certain distance, with several BVs driving on the right lane. If the CAV is very conservative and keeps a long safety distance with surrounding vehicles, it may fail to complete the lane-change task before the freeway exit. In such case, the vehicle may pass the safety evaluation but fail in the functionality evaluation.
Similar to safety evaluation, the functionality of a CAV can be evaluated by the failure rates of the CAV in completing certain driving tasks with different environment settings and BVs' trajectories.

We believe both safety and functionality are critical for CAV evaluation at the current technology maturity level. Unless a CAV can safely complete all driving tasks without human interventions, it may not be accepted by the general public.

For higher level requirements, mobility and rider's comfort should also be considered into the evaluation scope. Mobility is utilized to measure the travel efficiency in completing a series of driving tasks, while rider's comfort measures the physical and psychological feeling of passengers. Case studies of these two metrics will be investigated in future work.

\subsection{Performance Index Estimation}
Quantitative indices are designed to measure the performance metrics, e.g., the accident rate for safety performance and the failure rate for functionality performance. Here we denote the event of interest (e.g., accident) as $A$, and the occurrence probability of $A$ (e.g., accident rate) in the ODD is denoted as $P(A|\theta)$. 

In essence, on-road test is to estimate the performance indices of a CAV driving in the real world. For the cut-in example, if a test CAV drives on-road, experiences $n$ cut-in scenarios, and has $m$ accident events, the accident rate can be estimated by
\begin{eqnarray}
P(A|\theta) \approx \frac{m}{n}.
\end{eqnarray}
The theoretical justification is provided as follow. Assuming that the experienced cut-in scenarios follow the distribution of $P(x|\theta)$, i.e., $x_i \sim P(x|\theta)$, $i=1, \cdots, n$,  we can estimate the index as
\begin{eqnarray}
P(A|\theta) &&= \sum_{x\in \B{X}} P(A|x,\theta) P(x|\theta), \nonumber \\
&&\approx \frac{1}{n} \sum_{i=1}^{n} P(A|x_i,\theta), x_i \sim P(x|\theta),\\
&&\approx \frac{m}{n}, \nonumber
\end{eqnarray}
where the last two equivalences are derived by Monte Carlo theory  \cite{hammersley1964general}. As proved in \cite{kalra2016driving}, however, because the accident is a rare event, the required number of tests $n$ is intolerably large for reasonable estimation accuracy.

To improve the estimation efficiency, the importance sampling technique was introduced by \cite{zhao2017accelerated}. If an importance function $q(x)$ is properly constructed as
\begin{eqnarray}
&q(x) \in [0,1], \sum_{x\in\B{X}} q(x) = 1, P(x|\theta)>0 \Rightarrow q(x)>0,
\end{eqnarray}
and testing scenarios are sampled via the importance function, the index could be estimated by
\begin{eqnarray}
\label{eq_Index}
P(A|\theta) &&= \sum_{x\in \B{X}} P(A|x,\theta) P(x|\theta), \nonumber \\
&&= \sum_{x\in \B{X}} \frac{P(A|x,\theta) P(x|\theta)}{q(x)}q(x),  \\
&&\approx \frac{1}{n} \sum_{i=1}^{n} \frac{P(x_i|\theta)}{q(x_i)}P(A|x_i,\theta), x_i \sim q(x). \nonumber 
\end{eqnarray}
If the importance function $q(x)$ can assign higher probability for  critical scenarios, then more critical scenarios will be chosen during the test process. As a result, the required number of tests can be reduced, i.e., the evaluation method becomes more efficient. Zhao et al. \cite{zhao2017accelerated} has shown that a properly constructed importance function would significantly improve the safety evaluation efficiency for a low-dimensional scenario. For complex scenarios, however, the construction of a proper importance function still remains a problem.

\subsection{Objective of Testing Scenario Library Generation}
The objective of generating a testing scenario library is to properly construct the importance function $q(x)$, which can improve the estimation efficiency of Eq. (\ref{eq_Index}). If we can properly assign an importance value to each scenario, then those scenarios with importance value exceeding a threshold will be included in the testing scenario library. In this paper, the importance of a scenario is defined as a criticality measure, which is introduced in the next section.  

\subsection{Assumptions Made for TSLG}
The following assumptions are generally applied in the CAV tests, and both of them are mild. 
\begin{myAssume}
	\label{Assume_Rare}
	Testing CAVs are well-developed so that the event of interest $A$ is a rare event on-road. 
\end{myAssume}
\begin{myAssume}
	\label{Assume_Common}
	Testing CAVs share some ``generic features'' of behaviors.
\end{myAssume}

 Different types of CAVs may have generic features as well as unique features brought by their own manufacturers. The generic features capture fundamental functions of a well-developed vehicle behavior, e.g., keep safe distances and interact safely with surrounding vehicles. Similar to human drivers, where different drivers have different driving habits, generic features exist among all drivers.

\section{Testing Scenario Library Generation}
As we discussed above, the key to the testing scenario library generation (TSLG) problem is to compute the criticality value for each scenario. In this paper, a new criticality definition is proposed as a combination of exposure frequency and maneuver challenge. The exposure frequency can be estimated by using naturalistic driving data (NDD). To measure the maneuver challenge, a surrogate model (SM) of CAV is constructed. To reduce the computational complexity, an optimization-based method is applied to search for critical scenarios. An illustration of the proposed method for TSLG is shown in Fig. \ref{fig_LGMethod}.

\begin{figure}[h!]
	\centering
	\includegraphics[width=0.5\textwidth]{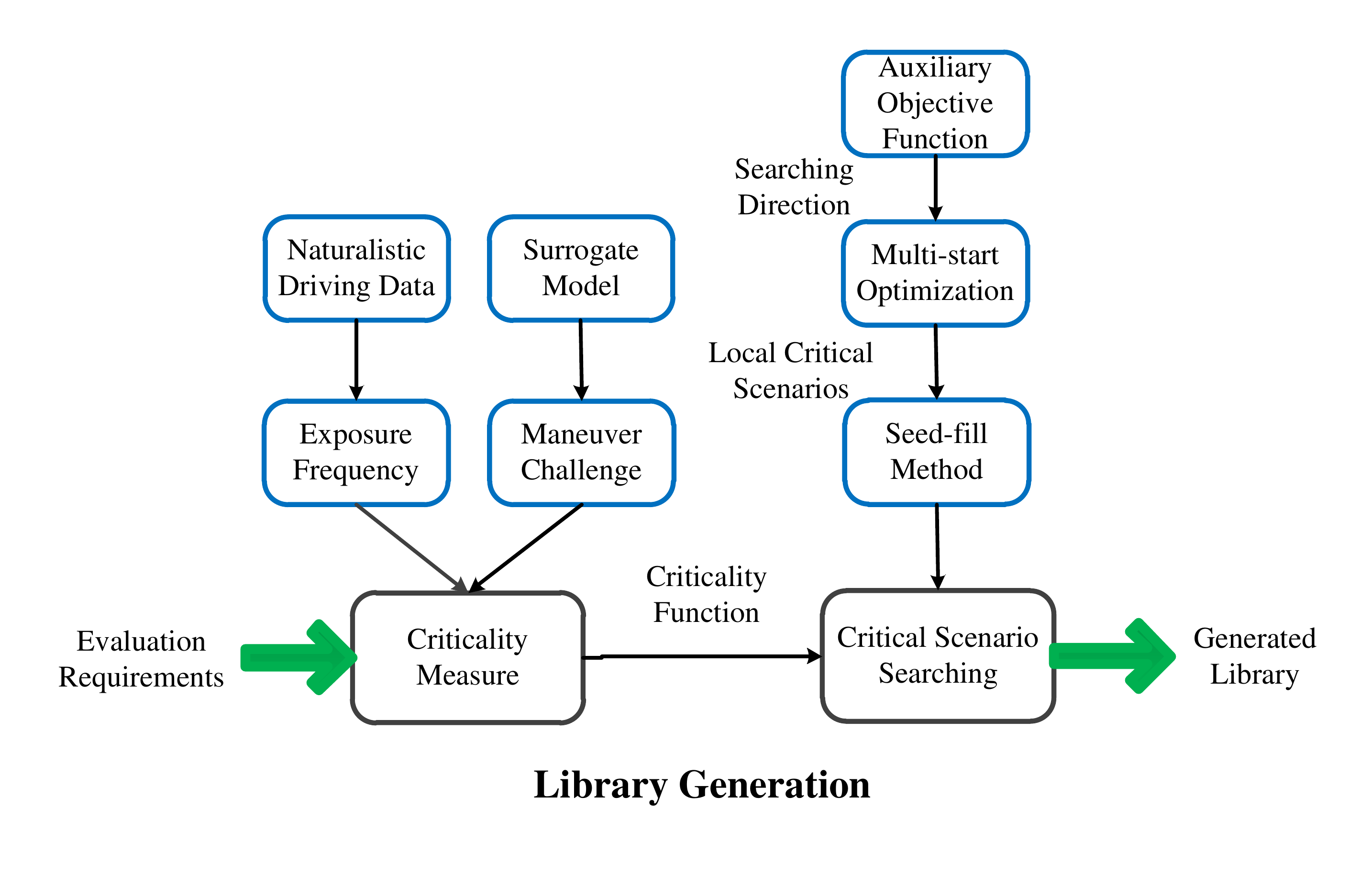}
	\caption{An illustration of the proposed method for TSLG.} 
	\label{fig_LGMethod}
\end{figure}

\subsection{Definition of Criticality }
\label{sec_lib_three}
The criticality of a scenario measures its importance in the evaluation of a performance metric. In ISO 26262 \cite{iso26262}, the risk assessment of a scenario was defined  as a combination of severity of injuries, exposure classification, and controllability classification. The exposure classification denotes the relative expected exposure frequency of the scenario where the injury can possibly happen. The controllability classification denotes the relative likelihood that the driver can act to prevent the injury.

Inspired by the concepts of the risk assessment, we define the criticality of scenarios as
\begin{eqnarray}
\label{eq_Value}
V(x|\theta) \overset{\rm def}{=} P(S|x, \theta) P(x|\theta), 
\end{eqnarray}
where $S$ denotes the event of interest (e.g., accident) with a SM of CAV. The reason for the introduction of the SM is that, for the purpose of TSLG, we assume that the exact CAV behavior model is not available. Therefore we introduce SM to reflect some of the generic features of different CAVs (see Assumption \ref{Assume_Common}).  An ideal SM should be calibrated from actual CAV driving data similar to human driving model calibration \cite{ranney1994models}. At the current stage, however, there is very little open CAV data available for public research.  Therefore, we propose to calibrate the SM based on the human driving data, i.e., NDD. This is a reasonable starting point as common behavioral  features of human drivers can serve as a natural baseline for CAV evaluation. Critical scenarios for human drivers are also meaningful testing scenarios for CAVs. In addition, many CAV algorithms are developed by imitating human driving behaviors, e.g., end-to-end learning method \cite{bojarski2016end}\cite{zhang2016query}. A ``human-like'' CAV can also improve safety in a mixed traffic condition, where CAVs and human-driven vehicles coexist on the roadway. A similar concept of ``roadmanship'' was recently proposed for CAV evaluation \cite{Rand2018safetymeasure}. Therefore, it is reasonable to use NDD to calibrate a SM in order to represent the generic features of CAVs.

The maneuver challenge ($P(S|x, \theta)$) measures the probability that a CAV encounters the event of interest in the scenario. The exposure frequency ($P(x|\theta)$) denotes the probability of the scenario occurring on-road. The justifications of this definition are theoretically proved regarding the evaluation accuracy and efficiency in Section \ref{Sec_Theory}. To calculate the criticality, $P(x|\theta)$ can be calculated according to NDD, and $P(S|x,\theta)$ is obtained by simulations of the SM. 

The definition also indicates that scenarios with higher occurrence probability in the real-world and higher maneuver challenge should have higher priority for CAV evaluation. Note although most of critical scenarios are rare, a portion of scenarios occur more frequently than others by orders of magnitude, e.g., $10^{-6}$ versus $10^{-9}$. This is fundamentally different from most existing studies, which usually overvalue the worst-case scenarios \cite{jung2007worst}\cite{zhao2018accelerated}. Taking an extreme example for conceptual explanation, the scenario that a meteor hitting a car is extremely dangerous but we cannot evaluate the performances of CAVs based on testing results from these extremely low frequent scenarios.

\subsection{Critical Scenario Searching}
The next problem is how to efficiently search the set of critical scenarios. The basic idea is to find local critical scenarios by optimization methods and then search their neighborhood scenarios. However, directly using the criticality function as the objective function is problematic. As discussed in Assumption \ref{Assume_Rare}, most scenarios are uncritical with zero criticality and zero gradient of criticality. Therefore, the criticality function provides little information of searching direction for critical scenarios. The optimization process degrades to a random sampling process, which is inefficient for complex scenarios. To address this issue, an auxiliary objective function is designed to guide searching directions. With the auxiliary objective function, the multi-start optimization method is applied to search the local critical scenarios, and the seed-fill method is applied to search neighborhood critical scenarios. The critical scenario searching method is summarized in Algorithm   \ref{alg_cri_sce_search}.

\begin{algorithm}
	\caption{Algorithm of critical scenario searching.}
	\label{alg_cri_sce_search}
	\LinesNumbered 
	\KwIn{Criticality function $V(x|\theta), x\in\B{X}$;}
	\KwOut{A library of critical scenarios $\Phi$;}
	{\bfseries Step 1}: Design an auxiliary objective function $J(x)$. \\
	{\bfseries Step 2}: Solve a multi-start optimization problem. Minimize $J(x)$, $x\in\B{X}$, with different initial starting points respectively, and obtain critical scenarios $x^*_i$, with $V(x^*_i|\theta)>\gamma$, $1\le i\le n^0$, where $n^0$ is the number of obtained local critical scenarios. The threshold $\gamma$ is obtained by Corollary \ref{Theo_threshold}.\\
	{\bfseries Step 3}: Seed-fill. Expand from the obtained local critical scenarios $x^*_i$, $1\le i\le n^0$, and find the set of critical scenarios, i.e., $\Phi = \{x\in \B{X}: V(x|\theta)>\gamma\}$. 
\end{algorithm}

First, an auxiliary objective function is designed as the combination of maneuver challenge and exposure frequency, similar to criticality definition. An example of the auxiliary objective function of the cut-in case for safety evaluation is shown as 
\begin{eqnarray}
\label{eq_Obj}
\min_x J(x) = \min_x \left(  mnpETTC(x) + w \times d(x, \Omega)  \right),
\end{eqnarray}
where $x = [R, \dot{R}]$ denotes the range and range rate at the cut-in moment. The first term is the minimal normalized positive enhanced time-to-collision (mnpETTC) during testing, which measures the danger level (i.e., maneuver challenge) of scenario $x$. The value of ETTC is calculated based on a surrogate car-following model as \cite{chen2016comparison} 
\begin{eqnarray}
\label{eq_ETTC}
ETTC(t) = \frac{-\dot{R}(t)-\sqrt{\dot{R}^2(t)-2u_r(t)R(t)}}{u_r(t)},
\end{eqnarray}
where $u_r$ is the relative acceleration. The minimal positive ETTC measures the most dangerous scene of a testing scenario. To make the metric comparable with exposure frequency, a normalization factor is applied. 
The second term is a normalized distance between the scenario and a high exposure frequency zone (i.e., $\Omega$) in NDD (e.g., 95\% percentile), in order to measure the exposure frequency of the scenario. $w$ is a weight parameter to balance the two terms. Because the auxiliary objective function is designed to approximate searching directions only, certain roughness of the designed function (e.g., caused by the value of $w$) is acceptable.

Second, a commonly used multi-start optimization method is applied to obtain a number of local critical scenarios. Specifically, multiple initial points are generated by space-filling methods (e.g., random sampling). After solving the optimization problem from each initial point as 
\begin{eqnarray}
\min_x J(x), x\in \B{X},
\end{eqnarray}
local critical scenarios are obtained, i.e., $x^*_i$, with $V(x^*_i|\theta)>\gamma$, $1\le i\le n^0$, where $n^0$ is the number of obtained local critical scenarios. The threshold $\gamma$ of critical scenarios is theoretically analyzed in Section \ref{Sec_Theory}. The number of initial points increases with the dimensions of the decision variables. Fortunately, the dimension of the decision variables can be greatly reduced by exploiting their specific structures, e.g., Markov property, and the searching method can be enhanced by RL techniques (see Part II \cite{feng2019testing} for examples).

Third, using the local critical scenarios as starting points, other critical scenarios are expanded by the seed-fill method. Seed-fill, also called flood-fill, is a basic method in computer graphics \cite{nosal2008flood} that determines the area connected to a given node in multi-dimensional arrays. The key idea is to exhaustively explore the critical points of unexplored space from the starting point outwards rather than all of the space \cite{kalisiak2006rrt}. The criticality function instead of the auxiliary objective function is used in this step, and the set of critical scenarios is defined as $\Phi = \{x\in \B{X}: V(x|\theta)>\gamma\}$.

\section{CAV Evaluation with the Library}
To test a CAV with the generated scenario library, three steps are involved including obtaining testing scenarios by sampling from the library, conducting CAV test with specified scenarios, and estimating performance indices from the testing results. An illustration of the CAV evaluation process is shown in Fig. \ref{fig_EvaMethod}.

\begin{figure}[h!]
	\centering
	\includegraphics[width=0.5\textwidth]{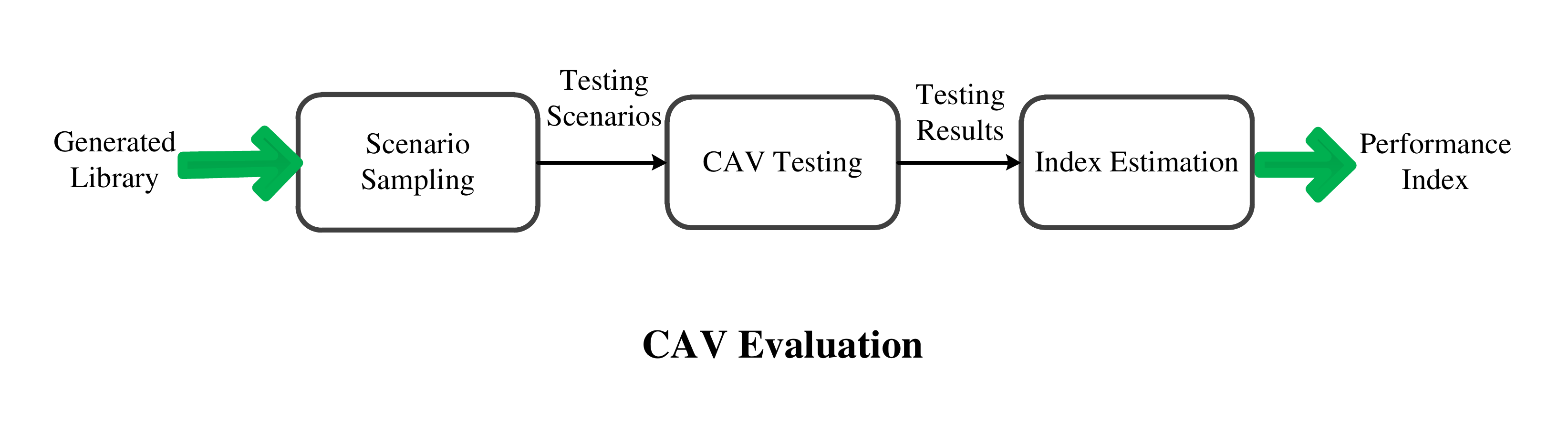}
	\caption{An illustration of the proposed method for the CAV evaluation process.} 
	\label{fig_EvaMethod}
\end{figure}

\subsection{Scenario Sampling}
The first step is to sample testing scenarios with a balance of exploitation and exploration. Recall that critical scenarios are obtained based on a SM, which usually has dissimilarity compared with the test CAV. Therefore, the generated library may miss some critical scenarios when testing a specific CAV. To address this issue, besides sampling scenarios from the library according to their criticality values (i.e., exploitation), the scenarios outside the library are also sampled with a small probability (i.e., exploration).

To better understand the trade-off between the exploitation and exploration, we compare the greedy sampling policy and $\epsilon$-greedy sampling policy. The greedy sampling policy \textbf{\emph{greedily}} exploits the scenarios in the library. By this policy, all testing scenarios are sampled based on the normalized criticality values. The $\epsilon$-greedy sampling behaves \textbf{\emph{greedily}} most of the time, but with small probability $\epsilon >0$, it selects scenarios randomly outside the library with equal probability (i.e., \textbf{\emph{exploration}}). This simple yet efficient method is commonly used for balancing exploitation and exploration \cite{sutton2011reinforcement}.

The testing probability distributions of scenarios with the two policies are derived as
\begin{eqnarray}
\label{eq_Prob_Sampling_1}
&\bar{P}_1(x_i|\theta) = \left\{
\begin{array}{ll}
V(x_i|\theta) / W, &x_i \in \Phi\\
0, &x_i \in \B{X} \backslash \Phi
\end{array}
\right. \\
\label{eq_Prob_Sampling_2}
&\bar{P}_2(x_i|\theta) = \left\{
\begin{array}{ll}
(1-\epsilon)V(x_i|\theta)/W, &x_i \in \Phi\\
\epsilon / (N(\B{X}) - N(\Phi)), &x_i \in \B{X} \backslash \Phi
\end{array}
\right.
\end{eqnarray}
respectively, where $N(\B{X})$ denotes the total number of feasible scenarios, and $W$ is a normalization factor as
\begin{eqnarray}
\label{eq_w}
W = \sum_{x_i \in \Phi} V(x_i|\theta).
\end{eqnarray}
The selection of $\epsilon$ is theoretically analyzed in Section \ref{Sec_Theory}.

From the perspective of importance sampling, the testing probability distributions in Eq. (\ref{eq_Prob_Sampling_1}-\ref{eq_Prob_Sampling_2}) essentially construct the importance function $q(x)$ in Eq. (\ref{eq_Index}). By involving the domain knowledge of CAVs and NDD, this construction method outperforms the general methods of importance sampling techniques (e.g., Cross Entropy method \cite{de2005tutorial}\cite{owen2013monte}). It can be applied for both low- and high- dimensional scenarios (see Part II \cite{feng2019testing}) and provides a feasible solution to progressively improve the importance function (see Theorem \ref{Theo_var}).

\subsection{CAV Testing}
The second step is to test the CAV with sampled scenarios by simulation or closed test facility. The initial conditions and maneuvers of BVs are determined by the sampled testing scenarios. The testing can be repeated easily by sampling different scenarios from the library. The total number of testing is determined by the required evaluation precision and confidence level \cite{zhao2017accelerated}\cite{wasserman2013all}\cite{ross2017introductory}. For example, at a confidence level $100(1-\alpha)\%$, to ensure the relative half-width of the estimation error is smaller than a predefined constant $\beta$, the number of tests needs to be larger than
\begin{eqnarray}
\label{eq_testRequire}
\frac{z_{\alpha}^2 }{\beta^2 \mu^2} \sigma^2,
\end{eqnarray}
where $z_{\alpha}$ is a constant,  and $\sigma, \mu= P(A|\theta)$ can be estimated by the variance and expectation of the testing results.



\subsection{Performance Index Estimating}
After the testing results are collected, the third step is to estimate the performance index value.
Substituting the constructed importance function into Eq. (\ref{eq_Index}), the index value can be estimated as
\begin{eqnarray}
\label{eq_Id_Em}
\hat{P}(A|\theta) \overset{\rm def}{=}  \frac{1}{n} \sum_{i=1}^{n} \frac{P(x_i|\theta)  }{\bar{P}(x_i|\theta)}P(A|x_i, \theta) ,
\end{eqnarray}
where $n$ denotes the total number of the sampled testing scenarios, $P(x_i|\theta)$ denotes the exposure frequency estimated from NDD, $\bar{P}(x_i|\theta)$ denotes the importance function, i.e., either $\bar{P}_1(x_i|\theta)$ or $\bar{P}_2(x_i|\theta)$ depending on the choice of the sampling policy, and $P(A|x_i, \theta)$ is estimated by the testing results.  The unbiasedness of Eq. (\ref{eq_Id_Em}) is proved in Theorem \ref{Theo_accuracy}.

\section{Theoretical Analysis}
\label{Sec_Theory}
In this section, the accuracy and efficiency of the proposed methods are validated by theoretical analysis, and choices of hyper-parameters, i.e., the threshold of critical scenarios and $\epsilon$, are discussed.

To simplify the notations, we omit the pre-determined parameters $\theta$ and define the following notations as
\begin{eqnarray}
\label{eq_notations}
&&f_A(x) = P(A|x, \theta),\nonumber \\
&&f_S(x) = P(S|x, \theta), \nonumber\\
&&p(x) = P(x|\theta), \nonumber\\
&&q_1(x) = {\bar{P}_1(x|\theta)}, \nonumber\\
&&q_2(x) = {\bar{P}_2(x|\theta)}, \\
&&\mu = P(A|\theta), \nonumber\\
&&\mu_S = P(S|\theta), \nonumber\\
&&\hat{\mu} = \hat{P}(A|\theta), \nonumber\\
&&W = \sum_{x_i \in \Phi} P(S|x_i, \varepsilon)P(x_i|\varepsilon). \nonumber
\end{eqnarray}

\subsection{Accuracy Analysis}
In this subsection, we prove that the proposed method can obtain unbiased (i.e., accurate) index estimation with $\epsilon$-greedy sampling policy. For greedy sampling policy, an additional condition is required for the unbiasedness.  
\begin{myTheo}
	\label{Theo_accuracy}
	The proposed evaluation method can obtain the unbiased performance index estimation, namely
	\begin{eqnarray}
	\label{eq_theo_accuracy}
	E(\hat{\mu}) = \mu,
	\end{eqnarray}
	under one of the following conditions:
	
	(1) with greedy sampling policy and $f_A(x) = 0, \forall x_i  \in \B{X} \backslash \Phi$;
	
	(2) with $\epsilon$-greedy sampling policy.
	
\end{myTheo}
\begin{proof}
	We first prove the theorem under the condition (2). 
	By the law of total probability, we obtain the right term of Eq. (\ref{eq_theo_accuracy}) as
	\begin{eqnarray}
	\mu = P(A|\theta) = \sum_{x_i \in \B{X}} P(A|x_i, \theta) P(x_i|\theta). \nonumber
	\end{eqnarray}
	Introducing the sampling probability $\bar{P}_2(x_i|\theta)$ as Eq. (\ref{eq_Prob_Sampling_2}), we obtain
	\begin{eqnarray}
	P(A|\theta) = \sum_{x_i \in \B{X}} \frac{P(A|x_i, \theta) P(x_i|\theta)}{\bar{P}_2(x_i|\theta)} \bar{P}_2(x_i|\theta). \nonumber
	\end{eqnarray}
	By Monte Carlo principle \cite{hammersley1964general}, if we sample $x_i \sim \bar{P}_2(x_i|\theta)$ for $n$ times, we have the estimation as
	\begin{eqnarray}
	\hat{\mu} = \hat{P}(A|\theta) = \frac{1}{n} \sum_{i=1}^{n} \frac{P(A|x_i, \theta)  P(x_i|\theta)  }{\bar{P}_2(x_i|\theta)}, \nonumber
	\end{eqnarray}
	as shown in Eq. (\ref{eq_Id_Em}). 
	As $\bar{P}_2(x_i|\theta)>0$ for all scenarios and the Central Limit Theorem \cite{rosenblatt1956central}, when $n$ is large, $\hat{P}(A|\theta)$ follows approximately the normal distribution with the mean
	\begin{eqnarray}
	E(\hat{\mu}) = \mu, \nonumber
	\end{eqnarray}
	which concludes the theorem under condition (2).
	
	For the theorem under condition (1), we have
	\begin{eqnarray}
	P(A|x_i, \theta) = 0, \forall x_i \in \B{X} \backslash \Phi \nonumber\\
	\bar{P}_1(x_i|\theta) = 0, \forall x_i \in \B{X} \backslash \Phi \nonumber
	\end{eqnarray}
	which indicates all scenarios outside $\Phi$ are uncritical. Therefore, the feasible set of decision variables can be changed from $\B{X}$ to $\Phi$, without loss of accuracy. Then, similar as the proof of the theorem under condition (2), the theorem under condition (1) can be proved.
\end{proof}

\begin{myRemark}
	The condition $f_A(x_i)=P(A|x_i, \theta)=0, \forall x_i \in \B{X} \backslash \Phi$, indicates that, for the test CAV, all scenarios outside the library satisfy $V(x_i|\theta)=0$. That is the reason why the greedy policy can be applied without loss of accuracy. However, considering the diversity of CAVs, this condition may not hold for real-world applications, so $\epsilon$-greedy policy is suggested.
\end{myRemark}

\subsection{Efficiency Analysis}
In this subsection, we prove that the estimation variance is small and even zero under certain conditions. Because the minimal number of tests is determined by the estimation variance (see Eq. (\ref{eq_testRequire})), the proposed method is proved to be efficient. 
\begin{myTheo}
	\label{Theo_var}
	The estimation variance is zero, i.e., $Var(\hat{\mu}) = \sigma^2 / n = 0$, under the following conditions
	
	(1) with the greedy sampling policy;
	
	(2) $f_A(x) = 0, \forall x_i \notin \Phi$;
	
	(3)  There exists a constant $k>0$ such that $f_A(x)=kf_S(x), \forall x \in \B{X}$.
\end{myTheo}
\begin{proof}
	According to the Monte Carlo method with importance sampling \cite{owen2013monte}, we obtain	the variance of the estimation as
	\begin{eqnarray}
	\label{eq_var_2}
	\sigma^2 &&= \sum_{x_i \in \Phi} \left(  \frac{f_A(x_i)p(x_i)}{q_1(x_i)} \right)^2 q_1(x_i) - \mu^2, \nonumber\\
	&&= \sum_{x_i \in \Phi}  \frac{ \left(     f_A(x_i)p(x_i) - \mu q_1(x_i) \right)^2  }{{q_1(x_i)}}, \nonumber \\
	&&=\sum_{x_i \in \Phi} \frac{p^2(x_i)}{q_1(x_i)} \left( f_A(x_i) - \mu \frac{q_1(x_i)}{p(x_i)}   \right)^2,
	\end{eqnarray}
	where the second equivalence is obtained by
	\begin{eqnarray}
	\sum_{x_i \in \Phi} q_1(x_i) = 1. \nonumber
	\end{eqnarray}
	By condition (2) and Eq. (\ref{eq_Value}), we have
	\begin{eqnarray}
	\label{eq_q}
	q_1(x_i) &&=P(S|\theta, x_i) P(x_i|\theta) / W, \nonumber \\
	&&=f_S(x_i) p(x_i) / W.
	\end{eqnarray}
	Substituting Eq. (\ref{eq_q}) into Eq. (\ref{eq_var_2}), we obtain
	\begin{eqnarray}
	\label{eq_var_3}
	&&\sigma^2 = \sum_{x_i \in \Phi}\frac{p^2(x_i)}{q_1(x_i)} \nonumber \\
	&& \times \left( f_A(x_i) - \frac{\mu}{W} f_S(x_i)  \right)^2.
	\end{eqnarray}
	Moreover, by the conditions (1-3), we have
	\begin{eqnarray}
	\label{eq_k}
	\frac{\mu}{W} &&= \frac{P(A|\theta)}{W}, \nonumber\\
	&&= \frac{\sum_{x_i \in \Phi}  P(A|x_i, \theta) P(x_i|\theta)   }{\sum_{x_i \in \Phi} P(S|x_i, \theta) P(x_i|\theta)}, \nonumber\\
	&&=k.
	\end{eqnarray}
	Substituting Eq. (\ref{eq_k}) into Eq. (\ref{eq_var_3}), we obtain
	\begin{eqnarray} 
	Var(\hat{\mu}) = \sigma^2 / n = 0, \nonumber
	\end{eqnarray}
	which concludes the theorem.
\end{proof}

\begin{myRemark}
	\label{Rmk_efficiency}
	As shown in Eq. (\ref{eq_testRequire}), if the estimation variance is zero, the minimal number of tests is one, which is ideal. Theorem \ref{Theo_var} shows strict conditions for the ideal results, which hold only if the SM is exactly the same as the test CAV. Since dissimilarity always exists between the SM and a specific CAV, the conditions are impossible to hold completely. Nevertheless, the theorem indicates that the source of the evaluation variance is the dissimilarity between the SM and the test CAV model (see Eq. (\ref{eq_var_3})). It also demonstrates that the evaluation efficiency can be further improved by mitigating the dissimilarity. Moreover, Theorem \ref{Theo_var} provides a foundation of determining hyper-parameters, i.e., the $\epsilon$ of the $\epsilon$-greedy sampling policy (Corollary \ref{Theo_var_epsilson}) and criticality threshold of critical scenarios (Corollary \ref{Theo_threshold}).
\end{myRemark}

\subsection{Choices of Hyper-parameters}
In this subsection, we provide methods to determine the hyper-parameters, i.e., $\epsilon$ and $\gamma$. 

\begin{myCor}
	\label{Theo_var_epsilson}
	The estimation variance with $\epsilon$-greedy sampling can be separated into two parts
	\begin{eqnarray}
	\sigma^2  &&= \sum_{x_i \notin \Phi} \frac{p^2(x_i)}{q_2(x_i)} \left( f_A(x_i) - \mu \frac{q_2(x_i)}{p(x_i)}   \right)^2 \nonumber \\
	&&+ \sum_{x_i \in \Phi}\frac{p^2(x_i)}{q_2(x_i)} \left( f_A(x_i) - \mu \frac{q_2(x_i)}{p(x_i)}   \right)^2, \nonumber
	\end{eqnarray}
	and the latter part is zero if  $\epsilon$ is chosen as
	\begin{eqnarray}
	\label{eq_epsilon}
	&&\epsilon = 1 - W / \mu_S, 
	\end{eqnarray}
	under the condition (3) in Theorem \ref{Theo_var}.
\end{myCor}

\begin{proof}
	Introduction of $\epsilon$ violates the condition (1) in Theorem \ref{Theo_var}. The variance in Eq. (\ref{eq_var_2}) changes as
	\begin{eqnarray}
	\sigma^2 &&= \sum_{x_i \in \B{X}} \frac{p^2(x_i)}{q_2(x_i)} \left( f_A(x_i) - \mu \frac{q_2(x_i)}{p(x_i)}   \right)^2, \nonumber \\
	&&= \sum_{x_i \notin \Phi} \frac{p^2(x_i)}{q_2(x_i)} \left( f_A(x_i) - \mu \frac{q_2(x_i)}{p(x_i)}   \right)^2 \nonumber \\
	&&+ \sum_{x_i \in \Phi}\frac{p^2(x_i)}{q_2(x_i)} \left( f_A(x_i) - \mu \frac{q_2(x_i)}{p(x_i)}   \right)^2. \nonumber
	\end{eqnarray}
	Denote the latter part as $\sigma^2_{\Phi}$ and we obtain
	\begin{eqnarray}
	\label{eq_var_eps}
	\sigma^2_{\Phi} &&= \sum_{x_i \in \Phi} \frac{p^2(x_i)}{q_2(x_i)} \times \\
	&&\left( f_A(x_i) - \mu \frac{(1-\epsilon)}{W} f_S(x_i)  \right)^2. \nonumber
	\end{eqnarray}
	Similar to Eq. (\ref{eq_k}), substituting Eq. (\ref{eq_epsilon}), we yield
	\begin{eqnarray}
	\label{eq_k_eps}
	\mu\frac{(1-\epsilon)}{W} &&= \frac{P(A|\theta)}{P(S|\theta)}, \nonumber\\
	&&=\frac{\sum_{x_i \in \B{X}}  P(A|x_i, \theta) P(x_i|\theta) }{\sum_{x_i \in \B{X}} P(S|x_i, \theta)P(x_i|\theta)}, \nonumber \\
	&&=k.
	\end{eqnarray}
	Substituting Eq. (\ref{eq_k_eps}) into Eq. (\ref{eq_var_eps}), we obtain
	\begin{eqnarray}
	\sigma^2_{\Phi}= 0, \nonumber
	\end{eqnarray}
	which concludes the theorem.
\end{proof}

\begin{myRemark}
The choice of $\epsilon$ as in Eq. (\ref{eq_epsilon}) will not increase the estimation variance for the scenarios in the library.
\end{myRemark}

\begin{myCor}
	\label{Theo_threshold}
	The estimation variance has an upper bound
	\begin{eqnarray}
	\label{eq_var_bound}
	\sigma^2 < \mu^2 \frac{(m-\epsilon)^2}{\epsilon}, 
	\end{eqnarray}
	if under the same conditions in Corollary 1 and the threshold of critical scenarios is determined as
	\begin{eqnarray}
	\label{eq_threshold}
	\gamma =  \frac{m \mu_S}{N(\B{X})-N(\Phi)},  
	\end{eqnarray}
	where $m\ge1$ is a constant.
\end{myCor}
\begin{proof}
	Note that $\Phi = \left\{x \in \B{X}: V(x|\theta)\ge\gamma\right\}$. By Eq. (\ref{eq_threshold}) and the condition (3) in Theorem \ref{Theo_var}, we obtain that for $x_i \notin \Phi$,
	\begin{eqnarray}
	P(x_i|A, \theta) &&= \frac{f_A(x_i)p(x_i)}{\mu}, \nonumber \\
	&&= \frac{kf_S(x_i)p(x_i)}{k \mu_S}, \nonumber \\
	&&= \frac{V(x_i|\theta)}{\mu_S}, \nonumber\\
	&&< \frac{m}{N(\B{X})-N(\Phi)}.
	\end{eqnarray}
	By Theorem 3 and Eq. (\ref{eq_Prob_Sampling_2}), we obtain
	\begin{eqnarray}
	\sigma^2 &&= \sum_{x_i \notin \Phi} \frac{p^2(x_i)}{q_2(x_i)} \left( f_A(x_i) - \mu \frac{q_2(x_i)}{p(x_i)}   \right)^2, \nonumber \\
	&&= \mu^2 \sum_{x_i \notin \Phi} \frac{1}{q_2(x_i)} \left( \frac{ f_A(x_i) p(x_i)}{\mu} - q_2(x_i)  \right)^2, \nonumber \\
	&&= \mu^2 \frac{N(\B{X})-N(\Phi)}{\epsilon} \nonumber \\
	&& \times \sum_{x_i \notin \Phi} \left( P(x_i|A, \theta) - \frac{\epsilon}{N(\B{X})-N(\Phi) } \right)^2.
	\end{eqnarray}
	By applying $m \ge 1 \ge \epsilon$ and properties of the quadratic function, we obtain the upper bound of the variance as
	\begin{eqnarray}
	\sigma^2 < \mu^2 \frac{(m-\epsilon)^2}{\epsilon},
	\end{eqnarray}
	which concludes the theorem.
\end{proof}

\begin{myRemark}
	Eq. (\ref{eq_var_bound}) shows the upper bound of the estimation variance. The choice of constant $m$, however,  has trade-offs, i.e., a larger $m$ decreases the size of the library but increases the upper bound of the estimation variance.  Eq. (\ref{eq_threshold}) shows that the determination of $\gamma$ can only be solved recursively as $ N(\Phi)$ is dependent on $\gamma$.  For practical applications, considering the rareness of critical scenarios, the threshold $\gamma$ can be relaxed as $ m \mu_S / N(\B{X})$.
\end{myRemark}

\section{Conclusions}
In this paper, we propose a unified framework for the testing scenario library generation (TSLG) problem for CAV evaluation. The framework can be used to generate testing scenario libraries for different ODD types,  performance metrics, and CAV models.

In this paper, the criticality of scenarios is defined as a combination of maneuver challenge and exposure frequency.  A multi-start optimization method is applied to search for the critical scenarios. To evaluate the maneuver challenge of scenarios, the surrogate model (SM) of CAVs is introduced, which contains the generic features of CAVs. Theoretical analysis is provided to ensure the accuracy and efficiency of the proposed testing method. It also demonstrates that the evaluation efficiency can be further improved by mitigating the dissimilarity between the SM and CAVs.

While this paper provides general framework and method to the TSLG problem, in Part II of this study  \cite{feng2019testing}, three case studies, including cut-in, car-following, and highway exit, will be investigated to demonstrate the proposed method. 
The proposed method is also enhanced using reinforcement learning technique for high-dimensional testing scenarios.

\appendices
\appendices

\bibliographystyle{IEEEtran}
\bibliography{IEEEexample}

\begin{IEEEbiography}[{\includegraphics[width=1in,height=1.25in,clip,keepaspectratio]{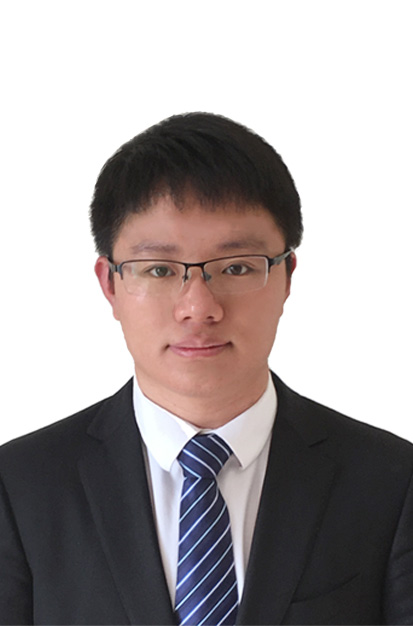}}]{Shuo Feng}
	received the bachelor’s degree and Ph.D. degree in Department of Automation from Tsinghua University, China, in 2014 and 2019. He was also a joint Ph.D. student in Civil and Environmental Engineering in University of Michigan, Ann Arbor, from 2017 to 2019. He is currently a postdoctoral researcher in Civil and Environmental Engineering in University of Michigan, Ann Arbor. His current research interests include connected and automated vehicle evaluation, mixed traffic control, and transportation data analysis.
\end{IEEEbiography}

\begin{IEEEbiography}[{\includegraphics[width=1in,height=1.25in,clip,keepaspectratio]{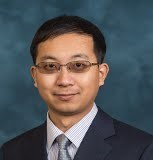}}]{Yiheng Feng}
	is currently an Assistant Research Scientist at University of Michigan Transportation Research Institute. He graduated from the University of Arizona with a Ph.D degree in Systems and Industrial Engineering in 2015. He has a Master degree from the Civil Engineering Department, University of Minnesota, Twin Cities in 2011. He also earned the B.S. and M.E. degree from the Department of Control Science and Engineering, Zhejiang University, Hangzhou, China in 2005 and 2007 respectively. His research interests include traffic signal systems control and security, and connected and automated vehicles testing and evaluation.
\end{IEEEbiography}

\begin{IEEEbiography}[{\includegraphics[width=1in,height=1.25in,clip,keepaspectratio]{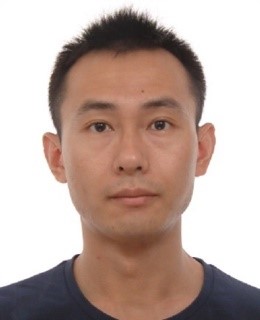}}]{Chunhui Yu}
	received the Ph.D. degree in transportation engineering from the College of Transportation Engineering in Tongji University, Shanghai, China, in 2018. He is currently an assistant researcher in the College of Transportation Engineering in Tongji University. Research interests include traffic signal optimization, vehicle trajectory optimization, and traffic management based on connected and automated vehicles.
\end{IEEEbiography}

\begin{IEEEbiography}[{\includegraphics[width=1in,height=1.25in,clip,keepaspectratio]{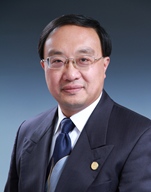}}]{Yi  Zhang}
	received   the   BS   degree   in1986   and   MS   degree   in   1988  from Tsinghua University in China, and earned  the  Ph.D.  degree  in  1995  from  the University  of  Strathclyde  in  UK.  He  is a   professor   in   the   control   science   and engineering  at  Tsinghua  University  with his  current  research  interests  focusing  on intelligent  transportation  systems. His  active  research  areas include  intelligent  vehicle-infrastructure  cooperative  systems, analysis  of  urban  transportation  systems,  urban  road  network management,  traffic  data  fusion  and  dissemination,  and  urban traffic control and management.   His research fields also cover the advanced control theory and applications, advanced detection and measurement, systems engineering, etc.
\end{IEEEbiography}

\begin{IEEEbiography}[{\includegraphics[width=1in,height=1.25in,clip,keepaspectratio]{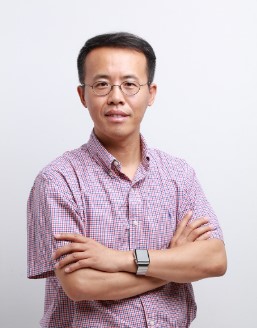}}]{Henry X. Liu}
	is a Professor of Civil and Environmental Engineering at the University of Michigan, Ann Arbor and a Research Professor of the University of Michigan Transportation Research Institute. He also directs the USDOT Region 5 Center for Connected and Automated Transportation. Dr. Liu received his Ph.D. degree in Civil and Environmental Engineering from the University of Wisconsin at Madison in 2000 and his Bachelor degree in Automotive Engineering from Tsinghua University in 1993. Dr. Liu's research interests focus on transportation network monitoring, modeling, and control, as well as mobility and safety applications with connected and automated vehicles. On these topics, he has published more than 100 refereed journal articles. Dr. Liu is the managing editor of Journal of Intelligent Transportation Systems and an associate editor of Transportation Research Part C. 
\end{IEEEbiography}

\end{document}